\documentclass[a4paper]{article}
\usepackage{amsmath,amssymb,amsthm,hyperref}
\usepackage{authblk}
\usepackage{enumerate}
\usepackage{todonotes}

\theoremstyle{plain}
\newtheorem{theorem}{Theorem}
\newtheorem{corollary}[theorem]{Corollary}

\newtheorem{lemma}[theorem]{Lemma}
\newtheorem{open}{Open Problem}
\theoremstyle{definition}
\newtheorem{definition}[theorem]{Definition}

\newtheorem{example}[theorem]{Example}

\newcommand{\edit}[1]{#1}

\title{On generalized Lyndon words}

\author[1]{Francesco Dolce}
\author[2]{Antonio Restivo}
\author[3]{Christophe Reutenauer}

\affil[1]{\small{
\edit{IRIF, Universit\'e Paris Diderot (France), \url{dolce@irif.fr}}
}}

\affil[2]{\small{
Dipartimento di Matematica e Informatica, Universit\`a degli Studi di Palermo (Italy), \url{antonio.restivo@unipa.it}
}}

\affil[3]{\small{
LaCIM, Universit\'e du Qu\'ebec \`a Montr\'eal (Qu\'ebec, Canada), \url{reutenauer.christophe@uqam.ca}
}}

\begin{document}

\maketitle

\begin{abstract}
A generalized lexicographical order on infinite words is defined by choosing for each position a total order on the alphabet.
This allows to define generalized Lyndon words.
Every word in the free monoid can be factorized in a unique way as a nonincreasing factorization of generalized Lyndon words.
We give new characterizations of the first and the last factor in this factorization as well as new characterization of generalized Lyndon words.
We also give more specific results on two special cases: the classical one and the one arising from the alternating lexicographical order.

\end{abstract}

\section{Introduction}
\label{sec:intro}

Let $A$ be a totally ordered alphabet.
A word $w$ is called a \emph{Lyndon word} if for each nontrivial factorization $w=uv$, one has $w < vu$ (here $<$ is the lexicographical order).
Lyndon words were introduced in~\cite{Lyndon}.
It is easy to see that this property can be expressed in an equivalent way using infinite words, namely $w^\omega < (vu)^\omega$ ( where $w^\omega = www \cdots$) for each nontrivial factorization $w =uv$.

A well-known theorem of Lyndon states that every finite word $w$ can be decomposed in a unique way as a nonincreasing product $w = \ell_1 \ell_2 \cdots \ell_n$ of Lyndon words.
This theorem, which is a combinatorial counterpart of the famous theorem of Poincaré-Birkhoff-Witt, provides an example of a factorization of the free monoid (see~\cite{L2}).
It has also many algorithmic applications and it may be computed in an efficient way. 
Indeed, Duval proposed in~\cite{D} a linear-time algorithms to compute it, while Apostolico and Crochemore proposed in~\cite{AC} a $O(\lg n)$-time parallel algorithm.
This factorization is also used in string processing algorithms (see~\cite{BDZZ}) and for the computation 
of runs in a word (see, e.g.,~\cite{CR}).

In this paper we consider a variant of this family of words: generalized Lyndon words.
These words were first introduced by the third author in~\cite{R2}.
Given a generalized order $<$, i.e., an order in which the comparison between two words depends on the length of their common prefix (see Section~\ref{sec:order} for the formal definition), a finite word $w$ is called a \emph{generalized Lyndon word} if for each nontrivial factorization $w=uv$ we have $w^\omega < (vu)^\omega$.

In this paper we present both new results and new proofs proofs of already published results concerning this family of words.
In~\cite{R2} it is proved that the family of generalized Lyndon words is a Hall set, and thus that they provide a factorization of the free monoid.
As a consequence, the associated Lie polynomials form a basis of the free Lie algebra (see~\cite{R1,R2}).
In the present paper, we give a new proof of this factorization theorem (Theorem~\ref{theo:factorization}) using only combinatorial techniques instead of the heavy machinery of Hall set theory.

Note that Nyldon words, introduced by Grinberg in~\cite{G}, also provide a factorization of the free monoid (see~\cite{CPS}), but they are not generalized Lyndon words.
Inverse Lyndon words introduced in~\cite{BDZZ} are not generalized Lyndon words neither, while anti-Lyndon words  (introduced in the same paper) with respect to a lexicographical order $<$ can be viewed as classical Lyndon words with respect to the order $\tilde<$ (see also Example~\ref{ex:lexinverse} later).

With our new combinatorial approach we are able, on one hand, to simplify several of the proofs from~\cite{R2} and, on the other hand, to obtain new interesting results.
In particular, we deduce a new characterization of the last factor of the unique nonincreasing factorization in Lyndon words (Corollary~\ref{cor:factorization3}).
We also simplify a result of~\cite{R2}, stating that generalized Lyndon words are characterised by their suffixes and show a new characterization by the prefixes (Theorem~\ref{theo:suffix}).
This last result is new even for classical Lyndon words.

Next, we focus on two particular cases of generalized orders: the classical and the alternating one.

In Theorem~\ref{theo:shortest} we give two new characterizations of the first factor of the nonincreasing factorization into classical Lyndon words.
For a different characterization of the first factor see also~\cite{D}~and~\cite[Lemma 7.14 (iii)]{R1}.

From Theorem~\ref{theo:shortest} we deduce a new proof of a result from Ufnarovskij (Corollary~\ref{cor:U}) which characterizes Lyndon words by their prefixes.

The second case we focus on, related to continued fractions, is given by Galois words.
These are generalized Lyndon words with respect to the {\em alternating lexicographical order} $<_{alt}$, that is the order comparing two words in an opposite way depending on the parity of the length of the common prefix (see Section~\ref{sec:galois} for the formal definition).
The link with continued fractions is that we have that $a_1 a_2 a_3 \cdots <_{alt} b_1 b_2 b_3 \cdots $ if and only if
$$
 a_1 + \frac{1}{a_2 + \frac{1}{a_3 + \frac{1}{\ddots}}}
 <
 b_1 + \frac{1}{b_2 + \frac{1}{b_3 + \frac{1}{\ddots}}}
$$
where $a_i, b_i \in \mathbb{N} \setminus \{ 0 \}$ for all $i$ (see also~\cite{R2}).

In Theorem~\ref{theo:p_even_or_odd} we give a characterization which generalized Ufnarovskij's Theorem to Galois words.
Moreover, we also characterize the first factor of the nonincreasing factorization in Galois words (Theorem~\ref{theo:shortest2}).
This is the analogue for Galois words of Theorem~\ref{theo:shortest}.

We conclude in Section~\ref{sec:conclusions} with some remarks and open problems. \\

\textbf{Acknowledgement.}
We thank the anonymous referee for his/her very detailed report with copious and useful suggestions.

\textbf{Dedication.}
Maurice Nivat has been one of the main figures of French school of Theoretical Computer Science.
Some of his contributions are in the field of combinatorics on words, as for instance the ones related to discrete geometry, where paths are coded by words (see~\cite{BN}).
Lyndon words, whose generalization we focus in this paper, have a significant role in discrete geometry (see, e.g.,~\cite{BLPR}).

We want to dedicate this paper, in this very international journal that he founded, to his memory.

\section{Definitions and notations}
\label{sec:def}
For undefined notation we refer to~\cite{L} and~\cite{L2}.
We denote by $A$ a finite alphabet, by $A^*$ the free monoid and by $A^+$ the free semigroup.
Elements of $A^*$ are called {\em words} and the identity element, denoted by $1$ is called the {\em empty word}.
We say that a word $u$ is a {\em factor} of the word $w$ if $w=xuy$ for some words $x,y$; $u$ is a {\em prefix} (resp. {\em suffix}) if $x=1$ (resp. $y=1$); it is {\em nontrivial} if $u \neq 1$ and {\em proper} if $u \neq w$. 
We say that $w=ps$ is a {\em nontrivial factorization} of $w$ whenever $p,s$ are both nonempty.
The length of a word $w = a_1 a_2 \cdots a_n$, where $a_i \in A$ for all $i$, is equal to $n$ and it is denoted by $|w|$.

A {\em period} of a word $a_1 a_2 \cdots a_n$ is a natural integer $p$ such that $a_i = a_{i+p}$ for any $i$ such that $i, i+p \in \{ 1, \ldots, n \}$; it is called {\em a nontrivial period} if  $0<p<n$.
A word having a nontrivial period is called {\em periodic}. 

We say that {\em $v$ is a fractional power of $u$} if $u = u_1 u_2$ and $v = u^k u_1$ for some nonnegative integer $k$.
In this case, one writes $v = u^r$, where $r = k + |u_1|/|u|$ is a positive rational.
Note that for $k=0$ (or $r<1$) this means that $v$ is a prefix of $u$.
Fractional powers are also known as \emph{sesquipowers} (see, e.g.,~\cite{PR}).

We say that the $v$ is {\em a strict fractional power of $u$} if $v$ is a fractional power of $u$ and, with the notations above, $k\geq 1$ or, equivalently, that $r\geq 1$.
In this case $u$ is a prefix of $v$.

\begin{example}
\label{ex:fractional}
Let $u = abcdef$.
Then $u^{2/3}=abcd$ and $u^{5/3}=abcdefabcd$.
The last one is, in particular, a strict fractional power of $u$.
\end{example}

We denote by $A^\omega$ the set of sequences over $A$, also called {\em infinite words}; such a sequence $(a_n)_{n \ge 1}$ is also written $a_1 a_2 a_3 \cdots$.
If $w$ is a (finite) word of length $n\geq 1$, $w^\omega$ denotes the infinite word having $w$ as a prefix and of period $n$.

We denote by  $A^\infty = A^* \cup A^\omega$ the set of finite and infinite words.

A {\em border} of a word $w$ of length $n$ is a word which is simultaneously a nontrivial proper prefix and suffix of $w$.
A word is called {\em unbordered} if it has no border.
It is well-known that a word has a border if and only if it is periodic.

Suppose that $u,v$ are finite nonempty words.
The following fact is well-known: one has $u^\omega = v^\omega$ if and only if $u,v$ are power of a common word, and this is true if and only if $u$ and $v$ commute (see, for instance,~\cite[Corollary 6.2.5]{L2}).

Observe also that if for two nonempty words $u,v$, one has $u^\omega \neq v^\omega$, then by Fine and Wilf theorem, their prefixes of length $|u|+|v|-\gcd(|u|,|v|)$ differ (see, for instance,~\cite{L}).

Given an order $<$ on the alphabet $A$, we can define the {\em lexicographic order} $<_{lex}$ (or simply $<$ when it is clear from the context) on $A^\infty$ in the following way : $u <_{lex} v$ if either $u$ is a proper prefix of $v$ (in which case $u$ must be in $A^*$) or we may write $u = pau'$, $v = pbv'$ for some words $p \in A^*$, $u',v' \in A^\infty$ and some letters $a,b \in A$ such that $a < b$.

\begin{definition}
\label{def:comparison}
Let $s, t$ be two distinct elements of $A^\omega$ such that we have a factorization $s = u_1 \cdots u_k s_0$ with $u_1, \cdots, u_k$ finite nonempty words and $s_0$ is an infinite word.
We say that {\em the comparison between $s$ and $t$ takes place within $u_k$} if $u_1 \cdots u_{k-1}$ is a prefix of $t$, but $u_1\cdots u_k$ is not.
If moreover $u_1, u_2, \ldots, u_k$ are letters we say that the comparison takes place {\em at position} $k$.
\end{definition}

Note that, when the comparison takes place within $u_k$, one may write $t = u_1 \cdots u_{k-1} u_k' t'$, for some $t' \in A^\omega$ and $u_k' \neq u_k$ such that $|u'_k| = |u_k|$.

\begin{lemma}
\label{lem:fractional}
Let $u,v$ be nonempty words such that $u^\omega\neq v^\omega$.
Then the comparison between $u^\omega$ and $v^\omega$ takes place within the first factor $v$ of $v^\omega$ if and only if $v$ is not a fractional power of $u$.
\end{lemma}
\begin{proof}
The comparison between the two infinite words takes place within the first $v$ if and only if the two prefixes of length $|v|$ of $u^\omega$ and $v^\omega$ are different.
The conclusion follows from the fact that $v$ is a fractional power of $u$ if and only if $v$ is a prefix of $u^\omega$.
\end{proof}

\section{Generalized lexicographical order}
\label{sec:order}

\begin{definition}
\label{def:genord}
For each $n \in \mathbb N \setminus \{ 0 \}$, let $<_n$ be a total order on $A$.
To the sequence $(<_n)_{n \ge 1}$ we associate a total order on $A^\infty$, that we still denote by $<$ when it is clear from the context, called {\em generalized lexicographical order}, as follows:
$u < v$ if either $u$ is a proper prefix of $v$ (in which case $u$ must be in $A^*$) or we can write $u = pas, v=pbt$ for some $p \in A^*$, some $s,t \in A^\infty$ and some letters $a,b \in A$ such that $a <_{|p|+1} b$.
\end{definition}

\begin{example}
\label{ex:p}
Let $<$ be the generalized order on $\{a, b\}^\infty$ defined by $b <_n a$ if $n$ is a prime number and $a <_n b$ otherwise.
Then we have $aba < abaaa < aab < bab < baab$ and $(ab)^\omega < a^\omega  < b^\omega < (ba)^\omega$.
\end{example}

Note that, as for the classical order, when $u$ is a prefix of $v$ we could have $u^\omega \nleq v^\omega$, as shown in the next example.

\begin{example}
\label{ex:prefix}
Let $A = \{a, b\}$, and $<$ as in Example~\ref{ex:p}.
The word $ab$ is a prefix of $aba$ but $(aba)^\omega < (ab)^\omega$.
\end{example}

Let us consider a generalized lexicographical order $<$ on $A^\infty$. 

\begin{lemma}
\label{lem:comparison2}
Let $s,t \in A^\omega$ be as in Definition \ref{def:comparison} (and the sentence following it).
Then $s < t$ (resp. $s > t$) implies $u_1 \cdots u_k' s' < u_1 \cdots u_k t'$ (resp. $u_1 \cdots u_k' s' > u_1\cdots u_kt'$) for any infinite words $s',t'$.
\end{lemma}

\begin{lemma}
\label{lem:p<s}
Let $u,v$ be nonempty finite words such that $u^\omega < v^\omega$ and let $x,y$ be two finite words.
Then
\begin{enumerate}
 \item[\rm (i)] if neither $u$ or $v$ is a prefix of the other, then $(ux)^\omega < (vy)^\omega$;
 \item[\rm (ii)] if $v$ is not a fractional power of $u$, then $(u^{k+1} x)^\omega < (v y)^\omega$, where $k$ is the largest integer such that $u^k$ is a prefix of $v$.
 In particular $u^\omega < (vy)^\omega$.
\end{enumerate}
\end{lemma}
\begin{proof} 
In case (i), the comparison between the two infinite words takes place within the prefix of length $\min(|u|,|v|)$.
Hence we conclude using Lemma~\ref{lem:comparison2}.

Suppose now that the hypothesis of (ii) holds.
Then we can write $u = u'au_1$ and $v = u^ku'bv_1$, with $u' \in A^*$, $a,b\in A$ such that $a\neq b$, and $u_1, v_1 \in A^*$.
Let $m = |u^k u'|$.
Since $u^\omega = u^k u' a u_1 u^\omega$ and since $u^\omega < v^\omega$, we have that $a <_{m+1} b$.
The two infinite words $u^\omega$ and $(u^{k+1}x)^\omega$ share the same prefix of length $m+1$, and the same do $v^\omega$ and $(vy)^\omega$.
Thus the comparison between between $(u^{k+1}x)^\omega$ and $(vy)^\omega$ takes place at position $m+1$.
Since $a <_{m+1} b$, we can conclude.
\end{proof}

We use several times the following observation: the opposite order $\tilde<$ of a generalized order $<$ is also a generalized lexicographical order, obtained by reversing all the orders $<_i$.

\begin{example}
\label{ex:lexinverse}
Let $<$ be the usual lexicographical order on $\{a, b\}$, that is such that $a <_i b$ for all $i \ge 1$.
Then $\tilde<$ is defined by $b~\tilde<_i~a$ for all $i \ge 1$.
\end{example}

\begin{example}
Let $<$ be the generalized order defined in Example~\ref{ex:p}.
Then we have $(aba)^\omega < (aab)^\omega < (bab)^\omega < (baa)^\omega$ and $(baa)^\omega~\tilde<~(bab)^\omega~\tilde<~(aab)^\omega~\tilde<~(aba)^\omega$.
\end{example}

Part of the following lemma is stated in~\cite{R2}.
\begin{lemma}
\label{lem:uuvv}
The following conditions are equivalent for nonempty words $u,v \in A^*$:
\begin{enumerate}
 \item[\rm (1)] $u^\omega < v^\omega$;
 \item[\rm (2)] $(uv)^\omega < v^\omega$;
 \item[\rm (3)] $u^\omega < (vu)^\omega$;
 \item[\rm (4)] $(uv)^\omega < (vu)^\omega$.
\end{enumerate}
\end{lemma}
\begin{proof}
The four conditions obtained by replacing in the lemma $<$ by $=$ are equivalent (see again~\cite[Corollary 6.2.5]{L2}).
We may therefore assume that none of these equalities holds.

Let us assume first that condition (1) holds and prove that the other conditions hold too.

If $v$ is not a fractional power of $u$, then conditions (2), (3) and (4) hold, by point (ii) of Lemma~\ref{lem:p<s}, with $x=u'bv_1$ and $y=1$ for case (2), $x=1$ and $y=u$ for case (3), and $x=u'bv_1$ and $y=u$ for case (4).

Let us suppose that $v$ is a strict fractional power of $u$.
We may therefore write $v = u^ju_1$, for some $j \ge 1$, and $u = u_1 u_2$.
Using the observation in the previous section, deduced from the Fine and Wilf theorem, we see that the prefixes of length $|u|+|v|$ of $u^\omega$ and $v^\omega$ are distinct, i.e., $u^{j+1} u_1 \ne u^j u_1 u$.
Since both $u^\omega$ and $(uv)^\omega$ begin by $u^{j+1} u_1$ and both $v^\omega$ and $(vu)^\omega$ begin by $u^j u_1 u$, and since the comparison in all four cases is done in the prefix of length $|u|+|v|$, we can conclude that (2), (3) and (4) hold.

Let us now consider the case when $v$ is a nonstrict fractional power of $u$, i.e., $v$ is a proper prefix of $u$.
That implies that either $u$ is not a fractional power of $v$ or $u$ is a strict fractional power of $v$.
Let us consider $\tilde<$ the opposite order of $<$.
Since $v^\omega~\tilde<~u^\omega$, from what we have seen above, we have that $(vu)^\omega~\tilde<~u^\omega$, $v^\omega~\tilde<~(uv)^\omega$ and $(vu)^\omega~\tilde<~(uv)^\omega$.
Thus, conditions (2), (3) and (4) hold.

Finally, let us suppose that the negation of (1) holds, that is that $u^\omega~\tilde<v~^\omega$ (remember that we supposing $u^\omega \ne v^\omega$).
Then, using the same reasoning as above, we have $(uv)^\omega~ \tilde<~v^\omega$, $u^\omega~\tilde<~(vu)^\omega$ and $(uv)^\omega~\tilde<~(vu)^\omega$, i.e., the negations of the three last conditions.
This shows that each of the condition (2), (3) or (4) implies (1).
\end{proof}

\section{Generalized Lyndon words}
\label{sec:words}
In this section we introduce generalized Lyndon words.

\begin{definition}
\label{def:genlyndon}
Given an alphabet $A$ and a generalized order $<$ on $A^\infty$
we say that a finite word $w \in A^+$ is a \emph{generalized Lyndon word} if for any nontrivial factorization $w = uv$ one has $w^\omega < (vu)^\omega$.
\end{definition}

\begin{example}
\label{ex:genlyn}
Let $A = \{ a,b \}$ and $<$ be the order defined in Example~\ref{ex:p}.
The word $w=abba$ is a generalized Lyndon word for the order $<$.
Indeed, one can easily check that $(abba)^\omega < (aabb)^\omega < (bbaa)^\omega < (baab)^\omega$.
\end{example}

\subsection{Characterization of generalized Lyndon words}

In the next theorem we give two characterizations of generalized Lyndon words.
Recall that a classical result due to Lyndon states that a word $w$ is a classical Lyndon word if and only if $w <_{lex} v$ for any nontrivial proper suffix of $w$ (see~\cite[Proposition 5.1.2]{L}).

The second part of the next result has already been proved in~\cite[Proposition 2.1]{R2}.
We give here a shorter proof.

\begin{theorem}
\label{theo:suffix}
Let us consider a generalized lexicographical order $<$ on $A^\infty$.
\begin{enumerate}
 \item[1.] A word $w$ is a generalized Lyndon word if and only if for any nontrivial factorization $w = uv$, one has $u^\omega < v^\omega$.\footnote{One may find on Wikipedia the following characterization (without proof nor references): $w$ is a classical Lyndon word if and only if for each nontrivial factorization $w=uv$ one has $u<v$. Our condition is not an extension of this condition to generalized Lyndon words. Indeed, if one take the usual order $a<b$, one has $b<ba$ but $b^\omega > (ba)^\omega$.}
 \item[2.] A word $w$ is a generalized Lyndon word if and only if for any nontrivial factorization $w = uv$, one has $w^\omega < v^\omega$.
\end{enumerate}
\end{theorem}
\begin{proof}
By definition, $w$ is a generalized Lyndon words if and only if for each nontrivial factorization $w = uv$, one has $(uv)^\omega < (vu)^\omega$.
By Lemma \ref{lem:uuvv}, this is equivalent both to $u^\omega < v^\omega$ and to $(uv)^\omega < v^\omega$, i.e. $w^\omega < v^\omega$.
\end{proof}

\begin{example}
Let $w, A$ and $<$ as in Example~\ref{ex:genlyn}.
Let us consider the nontrivial factorization $w=uv$ with $u=abb$ and $v=a$.
One has $(abb)^\omega < a^\omega$ and $(abba)^\omega < a^\omega$.
\end{example}

\subsection{Factorization into generalized Lyndon words}

The following result is already proved in \cite[Theorem 2.1]{R2} using the theory of Hall sets.
We give here an independent proof, especially for the uniqueness part, using only combinatorial arguments.

\begin{theorem}
\label{theo:factorization}
Each word in $A^*$ can be factorized in a unique way as a nonincreasing product of generalized Lyndon words.
\end{theorem}
\begin{proof}
Let us consider a generalized lexicographical order $<$ on $A^\infty$.

To prove the existence of such a nonincreasing factorization we follow the proof of~\cite[Corollary 2.2]{R2}.
Let $w\in A^+$ (if $w=1$ there is nothing to prove).
We define $z$ as the shortest among all nontrivial suffixes $s$ of $w$ such that $s^\omega$ is minimum.
By Theorem~\ref{theo:suffix}, $z$ is a generalized Lyndon word.
If $w=z$, we have found our factorization.
Otherwise, we can write $w = uz$ and, by induction on the length of $u$, we may assume that $u = \ell_1 \ell_2 \cdots \ell_n$, where the $\ell_i$ are generalized Lyndon words with $\ell_1^\omega \geq \ell_2^\omega \geq \ldots \geq \ell_n^\omega$.
Moreover, we have $\ell_n^\omega \geq z^\omega$.
Indeed, by construction of $z$ we have $(\ell_n z)^\omega \geq z^\omega$, and thus, using Lemma~\ref{lem:uuvv}, $\ell_n^\omega \geq z^\omega$.

Let us now prove the uniqueness of this factorization.
Suppose that we have $w = \ell_1 \ell_2 \cdots \ell_n$, where the $\ell_i$ are generalized Lyndon words with$\ell_1^\omega \geq \ell_2^\omega \geq \ldots \geq \ell_n^\omega$.
Let us show that $\ell_n$ is uniquely determined by the following condition: it is the shortest among all nontrivial suffixes $s$ of $w$ such that $s^\omega$ is minimum.
To prove this it is enough to show that if $s$ is a nontrivial proper suffix of $\ell_n$, then $\ell_n^\omega < s^\omega$; and if $s$ is a suffix of $w$ longer that $\ell_n$, then $\ell_n^\omega \leq s^\omega$. 
The first inequality follows from point 2 of Theorem \ref{theo:suffix} and the fact that $\ell_n$ is a generalized Lyndon word.
Suppose now that the second one is not true.
Thus there exists some $i$, with $2 \le i \le n$, and some factorization $\ell_{i-1} = uv$ with $v$ nonempty, such that $s = v \ell_i \cdots \ell_n$, and 
$$ (v \ell_i \cdots \ell_{n-1} \ell_n)^\omega <\ell_n^\omega. $$
From this last inequality and from Lemma~\ref{lem:uuvv} we deduce that 
$(v \ell_i \cdots \ell_{n-1})^\omega < \ell_n^\omega$.
Since $\ell_n^\omega \leq \ell_{n-1}^\omega$, we thus have
$$ (v \ell_i \cdots \ell_{n-1})^\omega < \ell_{n-1}^\omega.$$
Continuing recursively, we find that $(v \ell_i)^\omega < \ell_i^\omega$, therefore $(v\ell_i)^\omega < \ell_i^\omega$ and, since $\ell_i^\omega \leq \ell_{i-1}^\omega$, that $v^\omega < l_{i-1}^\omega$.
This gives us a contradiction to Theorem~\ref{theo:suffix}, since $\ell_{i-1} = uv$ is a generalized Lyndon word.
Thus $\ell_n$ is uniquely determined and, by proceeding recursively we prove the uniqueness of the factorization.
\end{proof}

From the proof of the previous theorem we obtain the two following results.

\begin{corollary}
\label{cor:factorization2}
Let $w = \ell_1 \ell_2 \cdots \ell_n$, with $\ell_i$ generalized Lyndon words such that $\ell_1^\omega \geq \ell_2^\omega \geq \ldots \geq \ell_n^\omega$.
Then $\ell_n$ is the shortest among all nontrivial suffixes $s$ of $w$ such that $s^\omega$ is minimum.
\end{corollary}

\begin{corollary}
\label{cor:factorization3}
With the same hypothesis as in Corollary~\ref{cor:factorization2}, we have that $\ell_n$ is the longest suffix of $w$ which is a generalized Lyndon word. 
\end{corollary}

Note that the this result is known for classical Lyndon words (see~\cite[Lemma 7.14 (ii)]{R1} and~\cite{D}).

\begin{proof}[Proof of Corollary~\ref{cor:factorization3}]
Indeed, if  there exists a suffix $s$ longer than $\ell_n$ which is a generalized Lyndon word, then, since $s$ has $\ell_n$ as a proper suffix, we would have $s^\omega < \ell_n^\omega$ by point 2 of Theorem~\ref{theo:suffix}, contradicting Corollary~\ref{cor:factorization2}.
\end{proof}

\begin{example}
\label{ex:pfact}
Let us consider the word $w = aabaabaabb$ with the order defined in Example~\ref{ex:p}.
The unique nonincreasing factorization of $w$ in generalized Lyndon words is $w = (a)(aba)(aba)(abb)$.
\end{example}

Note that every factor $\ell$ of the factorization in Lyndon words is \emph{primitive}, i.e., if $\ell = w^r$ with $w$ a finite word and $r$ an integer, then $r=1$ and $\ell=w$.

\section{Classical Lyndon words}
\label{sec:classical}

In this section, we take as generalized lexicographical order the usual lexicographical order $<_{lex}$, simply denoted by $<$.
Clearly, a generalized Lyndon word for this order is a usual one, since for two finite words of the same length, one has $u < v$ if and only if $u^\omega < v^\omega$. (see \cite[Theorem 8]{BMRRS}).

\subsection{Factorization into Lyndon words}
\label{sec:classical-factorization}

The nonincreasing factorization of a word into Lyndon words, as in Theorem~\ref{theo:factorization}, is the usual nonincreasing factorization into Lyndon words (see, for instance,~\cite{L}).

While at the end of Section~\ref{sec:words} we gave two characterizations of the last element of the factorization, here we focus on the first factor.
This result is motivated by point 1 of Theorem~\ref{theo:suffix}: the fact that a word $w$ is not a Lyndon word implies the existence of a prefix $u$
such that $u^\omega \geq v^\omega$, where $v$ is the corresponding suffix.
If one chooses the shortest prefix satisfying this property, this turns out to be the first factor in the Lyndon factorization.
In the same vein, it is motivated by Ufnarovskij's Theorem (Corollary~\ref{cor:U} below).

\begin{theorem}
\label{theo:shortest}
Let $w = \ell_1 \ell_2 \cdots \ell_n$ be the nonincreasing factorization into Lyndon words of a finite nonempty word $w$.
\begin{enumerate}
 \item[\rm 1.] The word $\ell_1$ is the shortest nontrivial prefix $p$ of $w$ such that, when writing $w=ps$, one has either $s=1$ or $p^\omega \ge s^\omega$.
 \item[\rm 2.] The word $\ell_1$ is the shortest nontrivial prefix $p$ of $w$ such that $p^\omega \geq w^\omega$.
\end{enumerate}
\end{theorem}

In order to prove Theorem~\ref{theo:shortest} we need a preliminary result which refines Lemma~\ref{lem:uuvv} in the case of the usual lexicographical order.

Note that, for any infinite words $s,t$ such that $s < t$, with $<$ the classical order, and for any finite word $w$, one has $ws < wt$.

\begin{lemma}
\label{lem:uuvv1}
Let $u,v$ be two nonempty words.
Then each of the two following conditions is equivalent to each of the four conditions in Lemma~\ref{lem:uuvv}:
\begin{enumerate}
 \item[\rm (5)] $u^\omega < (uv)^\omega$;
 \item[\rm (6)] $(vu)^\omega < v^\omega$.
\end{enumerate}
\end{lemma}
\begin{proof}
Condition (3) in Lemma~\ref{lem:uuvv} is equivalent to condition (5): indeed $ u^\omega < (vu)^\omega \Leftrightarrow uu^\omega < u(vu)^\omega \Leftrightarrow u^\omega < (uv)^\omega.$
Similarly, condition (2) is equivalent to condition (6): indeed,
$ (uv)^\omega < v^\omega \Leftrightarrow  v(uv)^\omega < vv^\omega \Leftrightarrow (vu)^\omega < v^\omega.$
\end{proof}

Note that the previous lemma implies that if $u^\omega < v^\omega$, then $$u^\omega < (uv)^\omega < (vu)^\omega < v^\omega,$$ a result proved by Bergman in~\cite[ Lemma 5.1]{B} (see also~\cite[p.34 and pp.101--102]{U}).

The following result is~\cite[Theorem 2, p.35]{U}.

\begin{corollary}[Ufnarovskij]
\label{cor:U}
A word $w$ is a Lyndon word if and only if for any nontrivial factorization $w=ps$, one has $p^\omega < w^\omega$.
\end{corollary}
\begin{proof}
It follows from point 1 of Theorem~\ref{theo:suffix} and from Lemma~\ref{lem:uuvv1}.
\end{proof}

\begin{example}
The word $w = aabab$ is a Lyndon word.
We have $a^\omega = (aa)^\omega < (aaba)^\omega < (aab)^\omega < w^\omega$.
\end{example}

\begin{corollary}
\label{cor:l_1}
If $\ell_1, \ell_2, \ldots, \ell_n$, with $n \ge 2$, are Lyndon words such that $\ell_1^\omega \ge \ell_2^\omega \ge \cdots \geq \ell_n^\omega$, then $\ell_1^\omega \ge (\ell_{2} \cdots \ell_n)^\omega$.
\end{corollary}
\begin{proof}
The case $n=2$, it is trivial.
Let consider the case $n \ge 3$.
By induction hypothesis we have $\ell_2^\omega \ge (\ell_3 \cdots \ell_n)^\omega$.
From Lemma~\ref{lem:uuvv1} it follows that $\ell_{2}^\omega \ge (\ell_{2}\cdots \ell_n)^\omega$.
Hence, $\ell_1^\omega \ge (\ell_{2} \cdots \ell_n)^\omega$.
\end{proof}

It is well-known that all (classical) Lyndon words are unbordered (see, for instance,~\cite{CK}).

\begin{proof}[Proof of Theorem~\ref{theo:shortest}]
Let us prove the first assertion.
When $n=1$, then $w = \ell_1$ is a Lyndon word and the result is true by point 1 of Theorem~\ref{theo:suffix}.

Suppose now that $n \geq 2$.
Then, by Corollary~\ref{cor:l_1}, we have $\ell_1^\omega \ge (\ell_{2} \cdots \ell_n)^\omega$.
Let $p$ be a nontrivial prefix of $w$ shorter then $\ell_1$.
Thus, we have a nontrivial factorization $\ell_1 = pq$ for some $q \ne 1$.
By Theorem~\ref{theo:suffix}, we know that $p^\omega < q^\omega$.
Since $\ell_1$ is unbordered, $q$ cannot be a fractional power of $p$.
Thus, by point (ii) of Lemma~\ref{lem:p<s}, one has $p^\omega < (q \ell_2 \cdots \ell_n)^\omega$, which prove the first part of the theorem.

The second assertion just follows from the first one.
Indeed, using Lemma~\ref{lem:uuvv1}, we have that if $s \ne 1$, then $p^\omega \ge s^\omega$ is equivalent to $p^\omega \ge (ps)^\omega$.
\end{proof}

\begin{example}
\label{ex:somegalois}
Let $w = ababaab$.
Its nonincreasing factorization into Lyndon words is $w = (ab)(ab)(aab)$.
One can check that $(ab)^\omega > w^\omega > (abaab)^\omega$ while $a^\omega < w^\omega < (babaab)^\omega$.
\end{example}

\section{Galois words}
\label{sec:galois}
In this section we consider a particular generalized lexicographical order.

\begin{definition}
\label{def:alt}
Let $<_1$ be an order on $A$.
The {\em alternating lexicographical order} $<_{alt}$ with respect to $<_1$ is the generalized lexicographical order defined by the sequence $(<_n)_{n \ge 1}$ with $<_n = <_1$ if $n$ is odd, and $<_n = \tilde<_1$ if $n$ is even.
\end{definition}

\begin{example}
\label{ex:alt}
Let us consider $<_1$ as the usual order on $\{a,b\}$.
Then one has $(ab)^\omega <_{alt} a^\omega <_{alt} b^\omega <_{alt} (ba)^\omega$.
\end{example}

This order is relevant when one orders real numbers through their continued fractions, see for example \cite[p.1-2]{R2}.
%We say that a word is {\em even} (resp. {\em odd}) if its length is.

The terminology in the following definition is justified in~\cite[p.2]{R2}.

\begin{definition}
\label{def:galois}
A {\em Galois word} is a generalized Lyndon word for an alternating lexicographical order.
\end{definition}

\begin{example}
Let us consider $<_1$ the usual order on $\{a,b,c\}$.
The following are Galois words: $b$, $ac$, $bc$, $aba$, $abb$, $abaa$, $acab$.
\end{example}

\subsection{Characterization of Galois words}
\label{sec:chargalois}
Similarly to what we saw in Section~\ref{sec:classical-factorization} for the classical order, for any infinite words $s,t$ such that $s <_{alt} t$, and any finite word $w$, one has $ws <_{alt} wt$ if $|w|$ is even, and $ws >_{alt} wt$ if $|w|$ is odd.

Symmetrically, when $ws <_{alt} wt$ one has $s <_{alt} t$ if $|w|$ is even and $s >_{alt} t$ if $|w|$ is odd.

\begin{example}
Let us consider the order of Example~\ref{ex:alt}.
One has $(ab) a^\omega <_{alt} (ab) b^\omega$ and $b a^\omega >_{alt} b b^\omega$.
\end{example}

Using the previous observation we can prove the next lemma using the same techniques as in Lemma~\ref{lem:uuvv1}.

\begin{lemma}
\label{lem:uuvv2}
Let $u,v$ be nonempty words.
Then each of the two following conditions is equivalent to each of the four conditions in Lemma~\ref{lem:uuvv} when considering the order $<_{alt}$:
\begin{itemize}
 \item[\rm (5)] $u^\omega <_{alt} (uv)^\omega$ if $|u|$ is even and $u^\omega >_{alt} (uv)^\omega$ if $|u|$ is odd;
 \item[\rm (6)] $(vu)^\omega <_{alt} v^\omega$ if $|v|$ is even and $(vu)^\omega >_{alt} v^\omega$ if $|v|$ is odd.
\end{itemize}
\end{lemma}
\begin{proof}
Let us first suppose that $|u|$ is even.
By the remark at the beginning of the section, one has $u^\omega <_{alt} (vu)^\omega \Leftrightarrow u^\omega = uu^\omega <_{alt} u(vu)^\omega = (uv)^\omega$.
Using the same remark we have, in the case $|u|$ is odd, $u^\omega <_{alt} (vu)^\omega \Leftrightarrow u^\omega = uu^\omega >_{alt} u(vu)^\omega = (uv)^\omega$.
Thus condition (3) of Lemma~\ref{lem:uuvv} is equivalent to condition (5).

Similarly, condition (2) of Lemma~\ref{lem:uuvv} is equivalent to condition (6).
Indeed, whenever $|v|$ is even one has $(uv)^\omega <_{alt} v^\omega \Leftrightarrow (vu)^\omega = v(uv)^\omega <_{alt} vv^\omega = v^\omega$, and whenever $|v|$ is odd one has $(uv)^\omega <_{alt} v^\omega \Leftrightarrow (vu)^\omega = v(uv)^\omega >_{alt} vv^\omega = v^\omega$.
\end{proof}

The following characterization of Galois words can be seen as a generalization of Ufnarovskij's Theorem (Corollary~\ref{cor:U}).

\begin{theorem}
\label{theo:p_even_or_odd}
A word $w$ is a Galois word if and only if for any nontrivial factorization $w = ps$, one has the following condition: 
$p^\omega <_{alt} w^\omega$ if $|p|$ is even and $p^\omega >_{alt} w^\omega$ if $|p|$ is odd.
\end{theorem}
\begin{proof}
The result immediately follow from point 1 in Theorem~\ref{theo:suffix} and from Lemma~\ref{lem:uuvv2}.
\end{proof}

\subsection{Factorization into Galois words}

Suppose that $w = g_1 g_2 \cdots g_n$ is the nonincreasing factorization of $w$ in Galois words.
We call {\em multiplicity} of $g_1$ the number
$m = \mbox{Card}\{i \mid g_i = g_1\}$.
In other words $w = g_1^mg_{i+1} \cdots g_n$, with $g_1^\omega >_{alt} g_{i+1}^\omega$.

The following result is a generalization of Theorem~\ref{theo:shortest} to Galois words.
This result is motivated by Theorems~\ref{theo:suffix} and~\ref{theo:p_even_or_odd}.

\begin{theorem}
\label{theo:shortest2}
Let $w = g_1 g_2 \cdots g_n$ with $g_i$ Galois words satisfying $g_1^\omega \geq_{alt} g_2^\omega \ge_{alt} \cdots \geq_{alt} g_n^\omega$.
Let $m$ be the multiplicity of $g_1$.
Let $p$ be the shortest nontrivial prefix of $w$ such that
\begin{equation}
\tag{$\star$}
p^\omega \geq_{alt} w^\omega \; \mbox{ if } |p| \mbox{ is even} \quad \mbox{ and } \quad p^\omega \leq_{alt} w^\omega \; \mbox{ if } |p| \mbox{ is odd.}
\label{eqn:*}
\end{equation}
Then
\begin{itemize}
 \item[\rm (i)] if $|g_1|$ is odd, $m$ is even, and $m<n$, then $p = g_1^2$;
 \item[\rm (ii)] otherwise, $p = g_1$.
\end{itemize}
\end{theorem}

Note that we can give an equivalent condition on $p$ in the previous statement.

\begin{lemma}
\label{lem:equiv}
Let $w=ps$ be a finite word, with $p \ne 1$.
Then $p$ satisfies condition $(\star)$ if and only if $p$ is such that one has either $s=1$ or $p^\omega \ge_{alt} s^\omega$.
\end{lemma}
\begin{proof}
From Lemma~\ref{lem:uuvv2} it follows that one has $p^\omega \geq_{alt} s^\omega \Leftrightarrow p^\omega \geq_{alt} w^\omega$ when $|p|$ is even and $p^\omega \geq_{alt} s^\omega \Leftrightarrow p^\omega \leq_{alt} w^\omega$ when $|p|$ is odd.
\end{proof}

In order to prove Theorem~\ref{theo:shortest2} we need several lemmata.

Recall, from Section~\ref{sec:classical-factorization} that classical Lyndon words have no border.
This is no more true for Galois words, as shown in the next lemma.

\begin{lemma}(\cite[ Proposition 3.1]{R2}
\label{lem:border}
If a Galois word has a border, then it has odd length.
\end{lemma}

\begin{lemma}
\label{lem:gprefixh}
Let $g,h$ be Galois words with $g^\omega <_{alt} h^\omega$ and $g$ a prefix of $h$.
Then $|g|$ is even.
\end{lemma}
\begin{proof}
Let $g$ be a prefix of $h$ with $|g|$ odd.
Then by Theorem~\ref{theo:p_even_or_odd}, one has $g^\omega >_{alt} h^\omega$.
\end{proof}

\begin{lemma}
\label{lem:hprefixg}
Let $g,h$ be Galois words with $g^\omega <_{alt} h^\omega$.
If  $g$ is a strict fractional power of $h$ then $|g|$ is even.
\end{lemma}
\begin{proof}
Let $k \ge 1$ and $h'$ a prefix of $h$ such that $g = h^k h'$.
The factorization $h = h'h''$ is not trivial since $h^\omega \neq g^\omega$.
Both $g^\omega$ and $h^\omega$ have $g$ as a prefix, and since $g^\omega <_{alt} h^\omega$, we have $g h' h'' s <_{alt} g h'' h' t$, where $s = h^{k-1} h' g^\omega$ and $t = h'' h^\omega$.

Let us suppose by contradiction that $|g|$ is odd.
By the remark at the beginning of Section~\ref{sec:chargalois}, we have $h'h'' s >_{alt} h''h' t$.
Since $|h' h''| = |h'' h'| = |h|$ but $h' h'' \ne h'' h'$, the comparison of the last inequality takes place within the prefix of length $|h|$.
Thus, by Lemma~\ref{lem:comparison2}, one has $h^\omega >_{alt} (h''h')^\omega$, which is impossible since $h$ is a Galois word.
\end{proof}

\begin{lemma}
\label{lem:g_1>or<}
Let $w = g_1 g_2 \cdots g_n$ with $g_i$ Galois words satisfying $g_1^\omega \ge_{alt} g_2^\omega \ge_{alt} \cdots \ge g_n^\omega$.
Let $m$ be the multiplicity of $g_1$ and assume that $m < n$ and that $n \geq 2$.
Then:
\begin{itemize}
 \item[\rm (i)] if $|g_1|$ is odd and $m$ is even, then $g_{1}^\omega <_{alt} (g_2 \cdots g_n)^\omega$;
 \item[\rm (ii)] otherwise, $g_{1}^\omega >_{alt} (g_2\cdots g_n)^\omega$.
\end{itemize}
\end{lemma}
\begin{proof} 
If $n=2$ then we have $m=1$ and thus $g_1^\omega >_{alt} g_2^\omega$., i.e., $m$ is odd and condition (ii) holds.

Suppose now that $n \geq 3$.
Let us first suppose that $|g_2|$ is even.
If $g_2 = \cdots = g_n$, then necessarily we have $g_1 \neq g_2$, since $m<n$.
Therefore, $g_1^\omega >_{alt} g_2^\omega = (g_2\ldots g_n)^\omega$ and $m=1$ is odd, so we are in case (ii).
If $g_2, \ldots, g_n$ are not all equal, we can argue by induction on $n$.
Thus since $|g_2|$ is even we have $g_2^\omega >_{alt} (g_3 \cdots g_n)^\omega$. 
By applying Lemmata~\ref{lem:uuvv} and~\ref{lem:uuvv2} we find that $g_2^\omega >_{alt} (g_2 g_3 \cdots g_n)^\omega$.
Finally, since $g_1^\omega \ge_{alt} g_2^\omega$, we have $g_1^ \omega >_{alt} (g_2 \cdots g_n)^\omega$; moreover, either $m$ is odd, or $m$ is even and then $g_1=g_2$ and $|g_1|$ is even, so that we are in case (ii).

Suppose now that $|g_2|$ is odd.
We assume first that $g_1 \neq g_2$, i.e., $m=1$.

We have $g_1^\omega >_{alt} g_2^\omega$. We show that $g_2$ is not a fractional power of $g_1$. Indeed, 
$g_2$ is not a prefix of $g_1$ by Lemma \ref{lem:gprefixh}; moreover, by Lemma \ref{lem:hprefixg}, $g_2$ is 
not a strict fractional power of $g_1$; since being a fractional power is equivalent to be a prefix, or a strict fractional power, we are done.

%Since $g_1^\omega >_{alt} g_2^\omega$, the comparison between $g_1^\omega$ and $g_2^\omega$ is within the first factor $g_2$.
%Indeed, $g_2$ is not a prefix of $g_1$ (by Lemma~\ref{lem:gprefixh}), and
%if $g_1$ is a prefix of $g_2$ then $g_2$ has not period $|g_1|$ (by Lemma~\ref{lem:hprefixg}).
Hence, by Lemma~\ref{lem:p<s} (ii) (applied to the opposite order), it follows that $g_1^\omega >_{alt} (g_2 \cdots g_n)^\omega$.

Finally, let us consider the case $g_1 = g_2$ and $|g_2|$ odd.
We have $m \geq 2$ and $w = g_1^m g_{m+1} \cdots g_n$, with $g_1^\omega >_{alt} g_{m+1}^\omega$.
By induction applied to $g_1 g_{m+1} \cdots g_n$ we have $g_1^\omega >_{alt} (g_{m+1} \cdots g_n)^\omega$, i.e., condition (ii) holds.
Since $|g_1| = |g_2|$ is odd, one deduces, by using Lemmata~\ref{lem:uuvv} and~\ref{lem:uuvv2}, that $g_1^\omega <_{alt} (g_1^{m-1} g_{m+1} \cdots g_n)^\omega$ when $m$ is even and $g_1^\omega >_{alt} (g_1^{m-1} g_{m+1} \cdots g_n)^\omega$ when $m$ is odd.
\end{proof}

An interesting consequence of the previous lemma is the following.

\begin{corollary}
\label{cor:g_1w}
Let $w = g_1 g_2 \cdots g_n$ with $g_i$ Galois words satisfying $g_1^\omega \ge_{alt} g_2^\omega \ge_{alt} \cdots \ge g_n^\omega$.
Let $m$ be the multiplicity of $g_1$.
One has $g_1^\omega >_{alt} w^\omega$ if $|g_1^m|$ is even and $g_1^\omega <_{alt} w^\omega$ if $|g_1^m|$ is odd.
\end{corollary}
\begin{proof}
This follows from Lemma \ref{lem:uuvv2} with $v = g_1^m$ and $u = (g_{m+1} \cdots g_n)$.
\end{proof}

We can now prove the main result of the section.

\begin{proof}[Proof of Theorem~\ref{theo:shortest2}]
Let us prove first that the two prefixes, $g_1$ for the case (ii) and $g_1^2$ for the case (i) satisfy condition $(\star)$.
If we are under the hypotheses of case (ii), i.e., if $|g_1|$ is even or $m$ is odd (the two conditions are not mutually exclusive), then $g_1^\omega >_{alt} (g_2 \cdots g_n)^\omega$ by Lemma~\ref{lem:g_1>or<}.
If we are under the hypotheses of case (i), then we have $n-1 \ge 2$ and the multiplicity of $g_1$ in $g_1g_3 \cdots g_n$ is odd.
Hence, applying Lemma~\ref{lem:g_1>or<} we find that $(g_1)^\omega = (g_1^2)^\omega >_{alt} (g_3 \cdots g_n)^\omega$.
In both cases the result follows from Lemma~\ref{lem:equiv}.

Let us now prove that any nontrivial proper prefix of $g_1$ in case (ii) and of $g_1^2$ in case (i) does not satisfy condition $(\star)$.

Let $g_1 = p t$ be a non trivial factorization of $g_1$ and $s = t g_2 \ldots g_n$.
By Theorem~\ref{theo:suffix}, we have $p^\omega <_{alt} t^\omega$.
If the comparison in the last inequality takes place within the first $t$ of $t^\omega$, then $p^\omega <_{alt} s^\omega$, and we can conclude by using Lemma~\ref{lem:equiv}.

Otherwise, by Lemma~\ref{lem:fractional}, $t$ is a fractional power of $p$, i.e., we can write $t = p^r$, $r \in \mathbb{Q} \setminus \mathbb{N}$ (since $g_1 = pt$ is primitive).
We claim that $|t|$ is odd.
Indeed, if we suppose that $r<1$, then $t$ is a prefix of $p$, hence of $g_1$, so that $|t|$ is odd by Lemma \ref{lem:border}.
If we suppose that $r>1$, then we can write $t = p^h p'$, with $h \geq 1$ and $p = p' p''$.
Thus $g_1 = (p' p'')^{h+1} p'$ and $p', p p'$ are both borders of $g_1$.
This implies by Lemma~\ref{lem:border} that $|p'|$ is odd and $|p|$ even, hence $|t|$ is odd.
This implies that $|p|$ is even if and only if $|g_1|$ is odd.

Let us suppose first that $|g_1|$ is odd.
Then $p^\omega <_{alt} g_1^\omega$ by Theorem~\ref{theo:p_even_or_odd}.
If we are in case (ii), we have $g_1^\omega \leq_{alt} 
w^\omega$ by Corollary~\ref{cor:g_1w}, hence $p^\omega <_{alt} w^\omega$, as we wanted to prove.
If we are in case (i), then $g_1 = g_2$.
If the comparison between $p^\omega$ and $g_1^\omega$ takes place within the first $g_1$ of $g_1^\omega$, we conclude that $p^\omega <_{alt} (g_1 g_2 \cdots g_n)^\omega = w^\omega$.
Otherwise, by Lemma~\ref{lem:fractional}, we can write $g_1 = p^k p_1$, with $k \geq 1$, and $p = p_1 p_2$.
Since $g_1 = (p_1 p_2)^k p_1$ is a Galois word, we have $((p_1 p_2)^k p_1)^\omega <_{alt} ((p_2 p_1)^k p_1)^\omega$.
Since $|(p_1 p_2)^k p_1| = |(p_2 p_1)^k p_1|$, the comparison in the previous (strict) inequality takes place in the prefix of their common length, hence in their prefix of length $|p| = |p_1 p_2| = |p_2 p_1|$.
Thus, $(p_1 p_2)s_0 <_{alt} (p_2 p_1)t_0$ for all infinite words $s_0, t_0$.

Since $g_2$ has $p_1 p_2$ as a prefix, we deduce that $(g_2 \cdots g_n g_1)^\omega <_{alt} (p_2 p_1)^\omega$.
Therefore, since $|g_1|$ is odd, we have
$$
 p^\omega = (p_1 p_2)^k p_1 (p_2 p_1)^\omega =
 g_1 (p_2 p_1)^\omega <_{alt}
 g_1 (g_2 \cdots g_n g_1)^\omega =
 w^\omega.
$$
Let us suppose now that $|g_1|$ is even (and thus $|p|$ is odd).
From Theorem~\ref{theo:p_even_or_odd} it follows that $p^\omega >_{alt} g_1^\omega$ and from Corollary~\ref{cor:g_1w} it follows that $g_1^\omega \ge_{alt} w^\omega$.
Thus $p^\omega >_{alt} w^\omega$.

We have proved that no nontrivial proper prefix of $g_1$ satisfy condition $(\star)$.
It remains to prove that, under the hypotheses of case (i), each proper prefix of $g_1^2$ of length at least $|g_1|$ does not satisfy condition $(\star)$.

Since we are in case (i), $|g_1|$ is odd,  $m$ is even and $g_1 \ne g_n$.
Using Lemma~\ref{lem:g_1>or<} we have that $g_1^\omega <_{alt} (g_2 \cdots g_n)^\omega$.
Thus it follows from Lemma~\ref{lem:equiv} that $g_1$ does not satisfy condition $(\star)$.

Let now consider $p = g_1 q$ with $g_1 = qt$ a nontrivial factorization of $g_1$.
If $|q|$ is even, and thus $|p|$ is odd, we have $q^\omega <_{alt} g_1^\omega$ by Theorem~\ref{theo:p_even_or_odd}.
Using Lemma~\ref{lem:uuvv2} and Corollary~\ref{cor:g_1w} we find
$$
 p^\omega = (g_1 q)^\omega >_{alt} g_1^\omega \ge_{alt} w^\omega.
$$
Finally, let us suppose that $|q|$ is odd, and thus $|p|$ and $|t|$ are even.
By Theorem~\ref{theo:suffix} we have $(qt)^\omega = g_1^\omega <_{alt} t^\omega$.
Since $t$ is not a prefix of $g_1$ (being of even length, it cannot be a border of $g_1$), the comparison is within the first $t$ of $t^\omega$.
Thus we have $(qtq)^\omega <_{alt} (t g_3 \cdots g_n qtq)^\omega$.
Since $|qtq|$ is even and $g_1 = g_2 = qt$, we deduce that
$$
 p^\omega = (g_1 q)^\omega = (qtq)^\omega <_{alt} (qtq t g_3 \cdots g_n)^\omega = w^\omega.
$$
Therefore, in both cases $p$ does not satisfy condition $(\star)$.
\end{proof}

\begin{example}
Let us consider the word $w = abbabbabaa$ using the alternating order of Example~\ref{ex:somegalois}.
The nonincreasing factorization of $w$ in Galois words is $w = (abb)(abb)(abaa)$.
One can check that $((abb)^2)^\omega >_{alt} w^\omega$, and that each nontrivial proper prefix of $(abb)^2$ does not satisfy condition $(\star)$ of Theorem~\ref{theo:shortest2}.
\end{example}

\section{Remarks and open problems}
\label{sec:conclusions}

Generalized Lyndon words, defined by using orders different than the usual lexicographical one, have different behaviors than classical Lyndon words.

For instance, we have seen that generalized Lyndon words are not, in general, unbordered, as is the case for classical Lyndon words.
Moreover, it is known that for each primitive word $w$, its unique conjugate which is a classical Lyndon word, is one of the elements in the nonincreasing factorization of $ww$ into Lyndon words (see, e.g. \cite[Section 7.4.1]{R1} and~\cite{L3} where are given algorithms to compute this unique conjugate).
This is no more true for generalized Lyndon words, as shown in the next example.

\begin{example}
\label{ex:counter}
Let us consider the primitive word $w = baa$.
Let us consider the alternating order on $\{ a,b\}^\infty$ with $a <_1 b$.
The unique Galois word conjugate to $w$ is $u = aba$.
The nonincreasing factorization of $ww$ into Galois words is $ww = (b)(a)(abaa)$.
\end{example}

In the classical case it is easy to show that we have a symmetric result of Corollary~\ref{cor:factorization2}, namely that $\ell_1^\omega$ is maximum along all $p^\omega$, with $p$ a nontrivial prefix of $w$.
This is not true for general orders, as shown in the next example.

\begin{example}
\label{ex:counter2}
Let us consider the word $w=abab$ and the alternating order of Example~\ref{ex:counter}.
Its nonincreasing factorization in Galois word is $w = (ab)(ab)$.
If we consider the nontrivial proper prefix $a$, we have $a^\omega >_{alt} (ab)^\omega$.
\end{example}

Moreover, using the same notation as before, it is not true in general that $\ell_1$ is the longest among all prefixes of $w$ which are generalized Lyndon words (this is true for classical Lyndon words, see, e.g., \cite[Lemma 7.14 (iii)]{R1}).

\begin{example}
\label{ex:counter3}
Let $w = abab$ as in Example~\ref{ex:counter2}.
The prefix $aba$ is longer than $ab$ and it is also a Galois word.
\end{example}

For classical Lyndon words, it is known that the unique nonincreasing factorization of a word in Lyndon words is also the factorization into Lyndon words which has the less number of factors\footnote{This follows easily from the property: if $u,v$ are Lyndon words and $u<_{lex}v$, then $uv$ is a Lyndon word, see \cite[Proposition 5.1.3]{L}}.
This is no more true for generalized Lyndon words, as shown in the next example.

\begin{example}
\label{ex:counter4}
Let us consider the alternating lexicographical order of Example~\ref{ex:counter}.
Then $w=(ab)(ab)(ab)$ is a word with its nonincreasing factorization in Galois words.
The word $w$ admits a shorter factorization into Galois words, namely $w= (ababa)(b)$.
\end{example}

In~\cite{D}, Duval shows that given a finite word $w$, it is possible to compute  in linear time its nonincreasing factorization into classical Lyndon words.

\begin{open}
\label{open:duval}
Generalize Duval's algorithm to generalized Lyndon words.
\end{open}

In the same paper, Duval also proposed a second algorithm generating all Lyndon words of length $\leq n$.
A consequence of this algorithm is that the number of Lyndon words of length at most $n$ is equal to the number of words of length $n$ that are prefixes of a Lyndon word plus $1$.
This property is no more true for a generalized order, as shown in the next example.

\begin{example}
\label{ex:duval}
Let us consider $A = \{ a,b,c \}$ with the usual order.
The only $6$ Lyndon words of length at most $2$ are $a,b,c,ab,ac$ and $bc$.
It is easy to check that there are exactly $5$ words of length $2$ which are prefixes of a Lyndon word, namely $aa, ab, ac, bb$ and $bc$.

Let us now consider $A$ with the alternating order given by $a <_1 b <_1 c$.
There are $6$ Galois words of length at most $2$ (namely $a,b,c, ab, ac$ and $bc$) but only $3$ words of length $2$ which are prefixes of Galois words (namely $ab, ac$ and $bc$).
\end{example}

Finally, all along the paper we only considered finite Lyndon words.
In~\cite{S} are defined infinite Lyndon words: these are the infinite words which have infinitely many prefixes that are (finite) Lyndon words. Then the authors of \cite{S} prove that $x$ is an infinite Lyndon word if and only if $x$ is smaller than any of its nontrivial proper suffixes (Proposition 2.2). They prove also that each infinite word $x$ is equal to a nondecreasing product of finite Lyndon words and perhaps one infinite one (Proposition 2.3). This means that either $x = \ell_1 \ell_2 \cdots \ell_n$, with $\ell_1, \ldots, \ell_{n-1}$ finite Lyndon words and $\ell_n$ an infinite one, or $x = \ell_1 \ell_2 \cdots$ is an infinite product of finite Lyndon words; in both cases, $\ell_1 \geq \ell_2 \geq \ldots$.

Thus, following \cite{BC} and \cite{S}, we say that an infinite word $x$ is  a {\em generalized infinite Lyndon word} if $x$ is smaller that any of its nontrivial proper suffixes.
It would be interesting to generalize this result to Galois words and other generalized Lyndon words.

\begin{open}
Prove that each infinite word can be factorized in a unique way as a nonincreasing product of finite and infinite generalized Lyndon words.
\end{open}

\end{document}